\documentclass[aps,pra,superscriptaddress,twocolumn]{revtex4-1}
\usepackage[utf8,latin1]{inputenc}
\usepackage[T1]{fontenc}     
\usepackage[british]{babel}
\usepackage[dvipsnames]{xcolor}
\usepackage[colorlinks=true,citecolor=Green,linkcolor=Red,urlcolor=Cyan,hyperindex]{hyperref}
\usepackage{graphicx}
\usepackage{epsfig}
\usepackage{color}
\usepackage{centernot}
\usepackage{lmodern}
\usepackage{mathpazo}
\usepackage[babel=true]{microtype}
\usepackage{amsmath,amssymb,amsthm}
\usepackage{mathtools}
\usepackage{bm}
\usepackage{enumerate}
\usepackage{sidecap}
\usepackage{enumitem}
\usepackage{listings}

\newcommand{\ket}[1]{\vert#1\rangle}
\newcommand{\bra}[1]{\langle#1\vert}

\newcommand{\1}{\mathbb{1}}

\newcommand{\ie}{\textit{i.e.}}

\newtheorem{theorem}{Theorem}

\newtheorem{coro}[theorem]{Corollary}
\newtheorem{prop}[theorem]{Proposition}
\newtheorem{example}{Example}
\newtheorem{conjecture}{Conjecture}
\newtheorem{definition}{Definition}

\newcommand{\cH}{\mathcal{H}}

\newcommand{\bM}{\textbf{M}}

\newcommand{\RR}{\mathbb{R}}
\newcommand{\PP}{\mathbb{P}}

\newcommand{\bA}{\textbf{A}}

\newcommand{\bQ}{\textbf{Q}}
\newcommand{\cA}{\mathcal{A}}
\newcommand{\cL}{\mathcal{L}}
\newcommand{\cB}{\mathcal{B}}
\newcommand{\cC}{\mathcal{C}}
\newcommand{\bB}{\textbf{B}}
\newcommand{\bR}{\textbf{R}}
\newcommand{\bC}{\textbf{C}}
\newcommand{\cR}{\mathcal{R}}

\newcommand{\bp}{\textbf{p}}

\newcommand{\beq}{\begin{equation}}
\newcommand{\eeq}{\end{equation}}

\begin{document}

\title{Joint measurability meets Birkhoff-von Neumann's theorem}
\author{Leonardo Guerini}
\email[]{guerini.leonardo@ictp-saifr.org}
\affiliation{{International Centre for Theoretical Physics - South American Institute for Fundamental Research \& Instituto de F\'isica Te\'orica - UNESP, R. Dr. Bento Teobaldo Ferraz 271, S\~ao Paulo, Brazil}}

\author{Alexandre Baraviera}
\email[]{baravi@mat.ufrgs.br}
\affiliation{{Instituto de Matem\'atica e Estat\'istica - Univerisdade Federal do Rio Grande do Sul, Av. Bento Gon\c{c}alves 9500, Porto Alegre, Brazil}}

\begin{abstract}
Quantum measurements can be interpreted as a generalisation of probability vectors, in which non-negative real numbers are replaced by positive semi-definite operators.
We extrapolate this analogy to define a generalisation of doubly stochastic matrices that we call doubly normalised tensors (DNTs), and formulate a corresponding version of Birkhoff-von Neumann's theorem, which states that permutations are the extremal points of the set of doubly stochastic matrices.
We prove that joint measurability arises naturally as a mathematical feature of DNTs in this context, needed to establish a characterisation similar to Birkhoff-von Neumann's.
Conversely, we also show that DNTs appear in a particular instance of a joint measurability problem, remarking the relevance of this property in general operator theory.

\end{abstract}

\maketitle

\section{Introduction}

Quantum theory is inherently probabilistic, in the sense that the result of a measurement on a quantum system cannot be predicted deterministically; we rather have to cope with a probability distribution over the set of possible outcomes \cite{nielsenchuang}.
A quantum system is  described by a Hilbert space and its states are given by density matrices, which are positive semi-definite operators of unit trace.
Due to this positivity and normalisation features, the density matrix can be interpreted as the quantum analogue of a probability vector, from which we extract information about the system.

On the other hand, quantum measurements are described by a collection of positive semi-definite operators that sum up to the identity.
Hence, the analogy between probability vectors and quantum measurements is even more direct, given the natural connections between non-negative real numbers and positive semi-definite Hermitian operators together with the association of 1 to the identity operator. 
Namely, the latter are an operator-version of the former, obtained by increasing the dimension of the entries.

In this work we explore this parallel, investigating the correspondent of standard features of probability vectors and their implications in terms of quantum theory.
More specifically, we aim at doubly stochastic matrices, matrices such that each column and row is a probability vector.
The set of such objects is convex, and an important characterisation of it is given by Birkhoff-von Neumann's (BvN's) theorem \cite{birkhoff}, which states that its extremal points are the permutation matrices.
Consequently, every doubly stochastic matrix is a convex combination of permutation matrices.

We introduce an operator-version of doubly stochastic matrices where probabilities are substituted by quantum measurements, which we call doubly normalised tensors.
Following BvN's theorem, our goal is to study the extremal points of such a set.
We extend the analogy between real numbers and operators to convex combinations, and find that a necessary condition for achieving a decomposition in terms of permutations 
given in terms of joint measurability \cite{heinosaari2016}, a property o measurements that plays a central role in many quantum information topics, such as Bell nonlocality \cite{brunner2014} and uncertainty relations \cite{busch2014}.
Nevertheless, we also show that not all doubly normalised tensors present this property, 
and hence a complete description of this set and its relations to permutations remains open.
We finish by presenting yet another connection to joint measurability: Starting from quantum theory and studying the plurality of joint measurements, we see the emergence of our generalised BvN's theorem arising from this context.

\section{Preliminaries}

A probability vector of $n$ components $\bp$ is given by
\begin{equation}\label{lasagna}
\bp=(p_1,\ldots,p_n)\in\RR^n;\ p_i\geq0,\ \sum_ip_i = 1.
\end{equation}
We denote the set of $n$-component probability vectors by $\mathbb{S}_n$.
An $n\times n$ doubly stochastic matrix is a matrix $D\in\RR^{n\times n}$ such that each column and each row is a probability vector,
\begin{subequations}\label{pizza}
\begin{align}
 &D=
 \begin{bmatrix}
 p_{11} & \ldots & p_{1n}\\
 \vdots & \ddots & \vdots\\
 p_{n1} & \ldots & p_{nn}
 \end{bmatrix};\\
\textbf{c}^1&=(p_{i1})_i,\ldots,\textbf{c}^n=(p_{in})_i\in\mathbb{S}_n, \\
\textbf{r}^1&=(p_{1j})_j,\ldots,\textbf{r}^n=(p_{nj})_j\in\mathbb{S}_n.
 \end{align}
\end{subequations}
The set of doubly stochastic matrices is convex.
Among the most important results on this topic lies Birkhoff-von Neumann's theorem \cite{birkhoff}, which states that the extremal points of this set are the permutation matrices.
This implies that every doubly stochastic matrix $D$ can be written as a convex combination of permutation matrices $\Pi_l$,
\begin{equation}\label{gnocchi}
D = \sum_{l=1}^{n!}q_l\Pi_l,
\end{equation}
where $\textbf{q}=(q_l)_l\in\mathbb{S}_{n!}$.
Although this decomposition involves in principle $n!$ terms,
it was shown in Ref. \cite{marcus1959} that only $(n-1)^2+1$ terms are sufficient.

In analogy to (\ref{lasagna}), a quantum measurement of $n$ outcomes acting on a Hilbert space $\cH$ is modelled by a positive-operator valued measure (POVM) $\bA$, described by
\begin{equation}
 \bA = (A_1,\ldots,A_n)\in\cL(\cH)^n;\ A_i\geq0,\ \sum_iA_i = I,
 \end{equation}
where $\cL(\cH)$ is the space of linear operators acting in $\cH$, $\geq$ is the partial order that define positive semi-definite operators and $I$ is the identity operator.
Hence, a POVM is an operator-version of a probability vector, obtained by enlarging the dimension of the entries.
We denote the set of $n$-outcome quantum measurements on $\cH$ by $\PP_n$.

We can consider now the main object of this work.
\begin{definition}
A \emph{doubly normalised tensor of positive semi-definite operators} (DNT, for short) is a tensor $\cA\in\cL(\cH)^{n\times n}$ which, in analogy to (\ref{pizza}), each element is positive semi-definite, and each column and each row sums up to the identity,
\begin{align}
 &\cA=
 \begin{bmatrix}
 A_{11} & \ldots & A_{1n}\\
 \vdots & \ddots & \vdots\\
 A_{n1} & \ldots & A_{nn}
 \end{bmatrix}.
 \end{align}
This is to say that each row $\bR^{(i)}$ and each column $\bC^{(j)}$ of $\cA$ is a quantum measurement,
\begin{subequations}\label{parmeggiana}
\begin{eqnarray}
 \bC^1=&(A_{i1})_i,\ldots,\bC^n=(A_{in})_i\in\PP_n, \\
 \bR^1=&(A_{1j})_j,\ldots,\bR^n=(A_{nj})_j\in\PP_n.
\end{eqnarray}
\end{subequations}
\end{definition}
Notice that we can write
\begin{equation}
\cA = \sum_{i,j=1}^n E_{i,j}\otimes A_{ij},
\end{equation}
where $E_{i,j} = \ket{i}\bra{j}$ is the $n\times n$ matrix with entries $e_{ab} = \delta_{a, i}\delta_{b,j}$, that belongs to the canonical basis of $\RR^{n\times n}$.
A DNT can be taken as an operator-version of doubly stochastic matrices, and will be denoted by $\cA=[A_{ij}]$.

\section{An operator-version of Birkhoff-von Neumann's theorem}

In view of Birkhoff-von Neumann's theorem and the decomposition (\ref{gnocchi}), we turn our attention to the extremal points of the set of DNTs.
Associating a DNT $\cA$ to the tuple of its row-measurements, it is known that this tuple is extremal if and only if each measurement is extremal \cite{guerini2018}.
Thus, applying the same argument to the columns, $\cA$ is extremal if and only each row and column is an extremal POVM.
But since each row $\bR^{(i)}$ and column $\bC^{(j)}$ share a common element (namely, $A_{ij}$), there are more correlations in $\cA$ that should allow for a more refined characterisation of its extremality.

A first (naive) conjecture would be that the extremal points of the set of DNTs are the operator-versions of permutations, such as
\begin{equation}\label{radiatori}
\begin{bmatrix}
 0 & 1\\
 1 & 0
 \end{bmatrix}
 \mapsto
 \begin{bmatrix}
 0 & I\\
 I & 0
 \end{bmatrix}.
\end{equation}
Nevertheless, convex combinations of DNTs like the one in the right-hand side above yield DNTs where all entries are proportional to $I$.
Clearly this does not recover the whole set, as for any operator $0\leq A\leq I$ we can construct the DNT
\begin{equation}
\begin{bmatrix}\label{tagliatelli}
 A & I-A\\
 I-A & A
 \end{bmatrix}.
\end{equation}

A second attempt is to extend the correspondence to operators also to convex combinations.
In other words, we can consider combinations of permutation matrices in which the convex weights are associated to operators that form a quantum measurement, attached to each term via tensor product.
We formalise this idea in the following way:
\begin{definition}
Consider  a set of operators $\bQ=(Q_l)\in\PP_{n!}$ that form a quantum measurement. We call
\begin{equation}\label{panini}
 \sum_{l=1}^{n!}\Pi_l\otimes Q_l,
\end{equation}
 a \emph{decomposition into permutation tensors}.
\end{definition}


We prove now that every combination of permutation tensors of the form (\ref{panini}) is a DNT.

\begin{prop}\label{fusilli}
If $\cB = \sum_{l=1}^{n!}\Pi_l\otimes Q_l$, where $\bQ=(Q_l)$ is a POVM and $\{\Pi_l\}$ is the set of $n$-dimensional permutation matrices, then $\cB$ is a DNT.
\end{prop}
\begin{proof}
Notice that the permutation $\Pi_l$ acting on the canonical basis can be written as\cite{stupidfootnote}
\begin{equation}\label{brie}
\Pi_l =\sum_i E_{\Pi_l(i),i}.
\end{equation}
We have
\begin{subequations}\begin{align}
\cB &= \sum_l\left(\sum_i E_{\Pi_l(i),i}\right)\otimes Q_l \\
&= \sum_{ab} E_{ab} \otimes \sum_{l:a=\Pi_l(b)} Q_l.
\end{align}\end{subequations}
Thus, defining
\beq\label{rondelli}
B_{ab}:=\sum_{l:a=\Pi_l(b)} Q_l,
\eeq
it satisfies $\cB=[B_{ab}]$ and $\ B_{ab}\geq 0$.
Also, for all $a,b$ we have
\beq
\sum_b B_{ab} = \sum_a B_{ab}= \sum_l Q_l = I,
\eeq
hence each row and column of $\cB$ is in $\PP_n$.
\end{proof}

The question left is whether tensors $\cB$ that admit a decomposition into permutation tensors are the only possible DNTs.
To address this question we need to introduce the well-known notion of \textit{joint measurability}~\cite{heinosaari2016}.

\begin{definition}
A set of $m$ quantum measurements $\{\bA^{(1)},\ldots,\bA^{(m)}\}\subset\PP_n$ is said to be \emph{jointly measurable} if there exists a so-called \textit{mother measurement} ${\bM}$, from which we can recover each measurement of the set by post-processing it, \ie,
\begin{equation}
A^{(i)}_j = \sum_k \mu(j|\bA^{(i)},k){M}_k.
\end{equation}
where
$\mu$ is a probability distribution conditioned on the measurement $\bA^{(i)}$ we wish to obtain and on the operator $M_k$ of ${\bM}$, and therefore satisfies $\mu(j|\bA^{(i)},k)\geq0$ and $\sum_j\mu(j|\bA^{(i)},k)=1,\ \forall i,k$.
\end{definition}
\noindent This expresses the fact that for any given quantum system, one can determine an outcome for each $\bA^{(i)}$ by performing ${\bM}$ and, depending on the outcome $k$ obtained, flip a coin $\mu(\cdot|\bA^{(i)},k)$.
Importantly, the joint measurability of a finite set of $d$-dimensional measurements can be computationally decided in an efficient way by means of semi-definite programming (SDP) \cite{wolf2009}.

We can now present the following proposition, which shows that any combination of permutation tensors has jointly measurable rows and columns.

\begin{prop}\label{capeletti}
Let $\cB =\sum_{l=1}^{n!}\Pi_l\otimes Q_l$ be a combination of permutation tensors.
Then $\cB$ is a DNT, and the set $\{\bR^{(1)},\ldots,\bR^{(n)},\bC^{(1)},\ldots,\bC^{(n)}\}$ of all row- and column-measurements is jointly measurable.
\end{prop}
\begin{proof}
$\cB$ is a DNT for Proposition \ref{fusilli}, hence we need only to prove that its rows/columns are both jointly measurable.

Defining the operators $B_{ab}$ as in (\ref{rondelli}), the row-measurements are given by $\bR^{(i)}=(B_{ij})_j$.
Then the coefficient measurement $\bQ$ is also a mother measurement for $\{\bR^{(i)}\}$, since for any $i,j$
\begin{equation}\label{farfalle}
B_{ij} = \sum_{l:i=\Pi_l(j)}Q_l = \sum_l{\delta_{i,\Pi_l(j)}Q_l},
\end{equation}
so we can recover the $j$-th element of the row-measurement $\bR^{(i)}$ by post-processing $\bQ$ with $\mu(j|\bR^{(i)},l) = \delta_{i,\Pi_l(j)}$.
Therefore the rows of $\cB$ are jointly measurable.

Extending the post-processing function for the columns as $\mu(i|\bC^{(j)},l)=\delta_{i,\Pi_l(j)}$, we can interpret (\ref{farfalle}) as obtaining the $i$-th element of $\bC^{(j)}$, and hence $\bQ$ is also a mother for the jointly measurable columns of $\cB$.
\end{proof}

Proposition \ref{capeletti} establishes that a necessary condition for a DNT to possess a decomposition into permutation matrices is that the row- and column-measurements are jointly measurable.
Crucially, in the proof we see that both the rows and the columns of the DNT not only can be obtained from a single measurement, but also admit symmetric post-processing functions, in the sense that $\mu(j|\bR^{(i)},l) = \mu(i|\bC^{(j)},l)$.
We now show that these conditions are also sufficient for ensuring such decomposition.

\begin{prop}\label{mozzarella}
Let $\cA=[A_{ij}]$ be an $n\times n$ DNT and $\cR = \{\bR^{(i)}=(A_{ij})_j\}_i,\ \cC = \{\bC^{(j)}=(A_{ij})_i\}_j$ its sets of row- and column-measurements, satisfying
\begin{itemize}
\item[(i)] $\cR \cup \cC$ is jointly measurable; and
\item[(ii)] the post-processing map is symmetric, $\mu(j|\bR^{(i)},k) = \mu(i|\bC^{(j)},k)$, for all $i,j,k$.
\end{itemize}
Then there exists a coefficient-measurement $\bQ\in\PP_{n!}$ such that
\begin{equation}\label{tortellini}
\cA = \sum_l \Pi_l\otimes Q_l.
\end{equation}
\end{prop}
\begin{proof}
Let ${\bM}$ be a mother measurement from which the symmetric post-processing map $\mu$ yields
\begin{equation}\label{pecorino}
A_{ij} = R^{(i)}_j = \sum_k \mu(j|\bR^{(i)},k)M_k = \sum_k \mu(i|\bC^{(j)},k)M_k = C^{(j)}_i.
\end{equation}
Notice that any asymmetric $\mu$ satisfies the relation above, but the assumed symmetry says that the two sums in (\ref{pecorino}) match term by term.
This ensures that
\beq
\sum_i\mu(j|\bR^{(i)},k) = \sum_i\mu(i|\bC^{(j)},k) = 1,\ \forall j.
\eeq
Also, by definition, for all $k$ we have
\beq
\sum_j\mu(j|\bR^{(i)},k) = 1,\ \forall i.
\eeq
Hence, for each $k$ we see that the matrix $(\mu(j|\bR^{(i)},k))_{ij} = \sum_{ij}\mu(j|\bR^{(i)},k)E_{ij}$ is doubly stochastic, which according to Birkhoff-von Neumann's theorem has a decomposition into permutation matrices,
\begin{equation}\label{fetuccine}
\sum_{i,j=1}^n \mu(j|\bR^{(i)},k)E_{ij} = \sum_l r^{(k)}_l\Pi_l,
\end{equation}
where $\textbf{r}^{(k)}=(r^{(k)}_l)_l\in\mathbb{S}_{n!}$ is a probability vector, for any $k$.

Therefore, using (\ref{pecorino}) and (\ref{fetuccine}) we obtain
\begin{subequations}\begin{align}
\cA =& \sum_{ij} E_{ij}\otimes A_{ij}\\
=& \sum_{ij} E_{ij}\otimes \left( \sum_k \mu(j|\bR^{(i)},k) {M}_k \right) \\
=& \sum_k \left( \sum_{ij} \mu(j|\bR^{(i)},k) E_{ij} \right) \otimes {M}_k \\
=& \sum_k \left( \sum_l r^{(k)}_l\Pi_l \right) \otimes {M}_k \\
=& \sum_l \Pi_l \otimes \left( \sum_k r^{(k)}_l {M}_k \right).
\end{align}\end{subequations}
Thus, defining $Q_l :=\sum_k r^{(k)}_l {M}_k$ we see that $Q_l \geq 0$ and $\sum_l Q_l = I$.
Therefore, $\bQ=(Q_l)$ is the coefficient-measurement that concludes the proof.
\end{proof}

Notice that if the response functions $\mu(j|\bR^{(i)},k)$ are deterministic, say, probability measures whose
support contains a unique point
(as the ones in the proof of Prop. \ref{capeletti}), then the left-hand side of (\ref{fetuccine}) is already a permutation matrix and the coefficient measurement $\bQ$ equals the mother measurement $\bM$.

Putting together Propositions \ref{capeletti} and \ref{mozzarella} we obtain the following characterisation.

\begin{theorem}\label{conchiglie}
A doubly normalised tensor of positive semi-definite operators admits a decomposition into permutation tensors if and only if its columns and rows are jointly measurable and admit symmetric post-processing functions.
\end{theorem}

A simpler hypothesis that can replace conditions $(i)-(ii)$ in Proposition \ref{mozzarella} refers to the linear independence of the operators of the mother measurement.
This allows us to assume joint measurability of only the rows (or, equivalently, of only the columns) of the DNT, relaxing $(i)$ but further specifying $(ii)$.

\begin{coro}
Let $\cA=[A_{ij}]$ be an $n\times n$ DNT satisfying
\begin{itemize}
\item[(i')] the set of rows $\cR = \{\bR^{(i)}=(A_{ij})_j\}_i$ of $\cA$ is jointly measurable; and
\item[(ii')] $\cR$ admits a mother measurement with linearly independent elements.
\end{itemize}
Then there exists a coefficient measurement $\bQ\in\PP_{n!}$ such that
\begin{equation}
\cA = \sum_l \Pi_l\otimes Q_l.
\end{equation}
\end{coro}

\begin{proof}
We will show that $(i')-(ii')$ imply $(i)-(ii)$ of Prop. \ref{mozzarella}.
Assume that for each $i,j$ we have
\begin{equation}
A_{ij} = \sum_k\mu(j|\bR^{(i)},k)M_k,
\end{equation}
for some mother measurement $\bM=(M_k)$ having linearly independent elements.
Since the columns of $\cA$ also sum to the identity, we have
\begin{equation}
I = \sum_i A_{ij} = \sum_k \left(\sum_i \mu(j|\bR^{(i)},k)\right)M_k.
\end{equation}
Since $\sum_k M_k =I$ by the normalisation of quantum measurements, linear independence implies that
\begin{equation}
\sum_i \mu(j|\bR^{(i)},l) = 1.
\end{equation}
Thus, the extended post-processing map $\mu(i|\bC^{(j)},k) := \mu(j|\bR^{(i)},k)$ (initially conditioned only on the rows) is well defined and yields the column-measurements of $\cA$, which therefore are jointly measurable with $\cR$.
Then $\bM$ and $\mu$ satisfy conditions $(i)-(ii)$ and we can apply Prop. \ref{mozzarella}.
\end{proof}

\section{Relaxing joint measurability}

In spite of characterising the set of DNTs that are decomposable into permutation tensors, Theorem \ref{conchiglie} does 
not describe general DNTs.
For that, we would need to show that all DNTs have jointly measurable row-measurements, besides the condition on the post-processing map.
Perhaps surprisingly, the next example shows that this is not the case, and not even the strong correlations between the rows of a DNT (given by the normalisation of its columns) are sufficient to enforce joint measurability.

\begin{example}
For $d=2$, consider the 3-outcome measurement $\bA = (A_1,A_2,A_3)$ whose elements are vertices of an equilateral triangle in the Bloch sphere representation,
\begin{subequations}
\begin{align}
A_1 &= \frac{I+\sigma_x}{3},\\ 
A_2 &= \frac{I-(\sigma_x-\sqrt{3}\sigma_z)/2}{3},\\ 
A_3 &= \frac{I-(\sigma_x+\sqrt{3}\sigma_z)/2}{3},
\end{align}
\end{subequations}
where $\sigma_x$ and $\sigma_z$ are Pauli matrices.
Writing $A_1' = \sigma_x/3$ and $A_1'' = I/3$, we have $A_1 = A_1'+A_1''$ and the following DNT,
\begin{align*}
 \cA=
 \begin{bmatrix}
 A_{1}'+A_1'' & A_2 & A_3\\
 A_2 & A_1'+A_3 & A_1''\\
 A_3 & A_1'' & A_1'+A_2
 \end{bmatrix},
 \end{align*}
where each entry is positive semi-definite.
However, one can check via semi-definite programming that these row-measurements are not jointly measurable \cite{wolf2009}, and use Proposition \ref{capeletti} to conclude that $\cA$ is not a combination of permutation tensors.

Nevertheless, we can still write
\begin{subequations}\begin{align}
\cA=&\begin{bmatrix}
1 & 0 & 0 \\ 0 & 1 & 0 \\ 0 & 0 & 1
\end{bmatrix}\otimes A_1'
+\begin{bmatrix}1&0&0\\0&0&1\\0&1&0\end{bmatrix} \otimes A_1''\\
+&\begin{bmatrix}0&1&0\\1&0&0\\0&0&1\end{bmatrix} \otimes A_2
+\begin{bmatrix}0&0&1\\0&1&0\\1&0&0\end{bmatrix} \otimes A_3,
\end{align}\end{subequations}
which is a combination of permutation matrices where not all coefficient-operators are positive semi-definite, given that $A_1'\ngeq0$.
\end{example}

Despite the fact that the rows of the DNT in the above example are not jointly measurable,
they still can be reconstructed by applying a post-processing map to the tuple of coefficients
$(A_1', A_1'', A_2, A_3)$, as the proof of Prop. \ref{capeletti} shows.
Since $A_1'\ngeq0$, this tuple is not a mother measurement, but it plays the same role as one.
Therefore, we will call it a \textit{pseudo-mother}.

Indeed, any set of measurements $\{\cB^{(1)},\ldots,\cB^{(n)}\}$ admits a  pseudo-mother like that;
if we no longer impose positive semi-definitiveness, we can simply consider the products $\widetilde{M}_{b_1\ldots b_n} = B^{(1)}_{b_1}\dots B^{(n)}_{b_n}$ and check that $\mu(j|\bB^{(i)},b_1\ldots b_n)=\delta_{b_i,j}$ is an appropriate post-processing map for it, since it satisfies
\begin{equation}
B^{(i)}_j = \sum_{b_1,\ldots,b_n}\widetilde{M}_{b_1\ldots b_n}\delta_{b_i,j}.
\end{equation}
Notice that $\widetilde{\bM}$ is still normalised, and many other such pseudo-mothers can be constructed.
For example, the order of the operators in its defining product can be arbitrary, as long as it is the same for each element of the pseudo-mother.

Hence, by dropping positive semi-definitiveness from the results in the last section, it is straightforward to obtain the following result.

\begin{theorem}\label{gorgonzolla}
Let $\cA=[A_{ij}]$ be an $n\times n$ DNT and $\cR = \{\bR^{(i)}=(A_{ij})_j\}_i,\ \cC = \{\bC^{(j)}=(A_{ij})_i\}_j$ its sets of row- and column-measurements.
Then the following are equivalent:
\begin{itemize}
\item[(a)] $\cR\cup\cC$ admits a pseudo-mother measurement $\widetilde\bM$ with a symmetric post-processing map, $\mu(j|\bR^{(i)},k) = \mu(i|\bC^{(j)},k)$ for all $i,j,k$;
\item[(b)] $\cA=\sum_l\Pi_l\otimes \widetilde Q_l$, where the operators $\widetilde Q_l$ are normalised, 
but are not necessarily positive semi-definite.
\end{itemize}
\end{theorem}

Although every set of measurements admits a pseudo-mother, we were not able to show that one with a symmetric post-processing can always be found, nor to present any counter-example to it.
Therefore, we leave it as an open question whether every DNT satisfies item $(a)$ of Theorem \ref{gorgonzolla}, and consequently can be characterised by the decomposition presented in item $(b)$.

\section{DNTs arising from a joint measurability problem}

Throughout this work, we presented our motivations to define DNTs and a generalisation of Birkhoff-von Neumann's theorem as purely mathematical, namely to further extend the clear parallel between probability vectors and POVMs.
We now show that the set of DNTs can be reached also by a quantum-theoretical path, more specifically in terms of joint measurability.

The property of joint measurability is based on the existence of a mother measurement, but another natural question refers to the uniqueness of such object \cite{guerini2018}.
Restricting the post-processing map to be deterministic, and therefore a marginalisation (which can be done without loss of generality \cite{ali2009}), we can display the mother measurement $\bM=(M_{ab})$ for a pair of POVMs $\bA=(A_1,\ldots,A_n), \bB=(B_1,\ldots,B_n)$ as a table,
\begin{eqnarray}
\begin{tabular}{ccc|c}
$M_{11}$ & $\cdots$ & $M_{1n}$ & $A_1$\\
$\vdots$ & $\ddots$ & $\vdots$ & $\vdots$\\
$M_{n1}$ & $\cdots$ & $M_{nn}$ & $A_n$\\
\hline
$B_1$ & $\cdots$ & $B_n$ &
\end{tabular},
\end{eqnarray}
emphasising the marginalisations $\sum_b{M_{ab}} = A_a$ and $\sum_a{M_{ab}}=B_b$.

Once we decide to study the plurality of mother measurements for a given pair, it is reasonable to start from the most basic specimen.
Taking $\bA=\bB=(I/n,\ldots,I/n)$, we guarantee that the pair is trivially joint measurable, for being both copies of the same POVM, and, on top of that, a trivial one.
However, this trivial case allows to see that the general mother measurement for this pair, upon rescaling all the operators by a factor of $n$ (the number of outcomes), yields a table
\begin{eqnarray}
\begin{tabular}{ccc|c}
$nM_{11}$ & $\cdots$ & $nM_{1n}$ & $I$\\
$\vdots$ & $\ddots$ & $\vdots$ & $\vdots$\\
$nM_{n1}$ & $\cdots$ & $nM_{nn}$ & $I$\\
\hline
$I$ & $\cdots$ & $I$ &
\end{tabular},
\end{eqnarray}
which has exactly the DNT features, \ie, it can be provided with a tensor structure, and its rows and columns form POVMs.

Thus, we see that a generalisation of BvN's theorem emerges not only from an operator-version of the original result, but also as the description of the set of mother measurements for perhaps the most trivial pair of POVMs that accepts multiple mothers.

\section{Discussion}

In this work we explored the analogy between probability vectors and POVMs and extrapolated it to investigate a generalised version of Birkhoff-von Neumann's theorem.
The richer structure of non-negative operators yields a discrepancy between the set of doubly normalised tensors (that generalise doubly stochastic matrices) and decompositions into permutation tensors (that correspond to convex combinations of permutations).
We showed that the latter can be perfectly described as the DNTs composed of jointly measurable POVMs with symmetric post-processing maps, remarking joint measurability -- a quantum-theoretical concept with a strongly operational motivation -- as an relevant mathematical property by itself.

On the other hand, the general set of DNTs remains to be characterised.
Our Theorem \ref{gorgonzolla} states that the DNTs arising from a symmetric post-processing can be written as affine combinations of permutation tensors; it is not clear whether such symmetry imposes a non-trivial constraint or, on the contrary, arbitrary DNTs satisfy this requirement.
This result is similar to a characterisation of unital quantum channels presented in Ref. \cite{mendl2009}, where it was proved that these objects are affine combinations of unitary channels.
The similarity is curious, given that unital quantum channels are also generalisations of doubly stochastic matrices in some sense.
In a broader context, the need for quasi-POVMs in our description of incompatible DNTs dialogues with the need for quasi-probability representations in quantum theory \cite{wigner1932,lutkenhaus1995}.

As mentioned along the text, in the literature on BvN's theorem is also posed the question of how many permutations are needed to describe an arbitrary doubly stochastic matrix.
It was shown that $(n-1)^2+1$ permutations were enough \cite{marcus1959}, and that although many different decompositions exist, the problem of computing the optimal one (with the minimal number of terms) is NP-complete \cite{dufosse2016}.
Here, Theorem \ref{conchiglie} establishes that any jointly measurable, symmetric DNT can be decomposed into permutations matrices, with coefficient operators that form a POVM.
Hence, an extremal DNT with these properties must correspond to an extremal coefficient POVM.
Extremal POVMs on dimension $d$ have at most $d^2$ non-null elements, thus by fixing the dimension of the underlying Hilbert space (which is independent of the size $n\times n$ of the tensor), we arrive at an uniform upper bound for the number of terms in the decomposition of extremal DNT's.

Finally, we recall that doubly stochastic matrices go together with the concept of majorisation of real vectors and Hermitian matrices \cite{marshall2011,watrous}, a celebrated connection with important applications to quantum information theory \cite{nielsen1999}.
It would be interesting to understand whether DNTs are associated to a similar notion of majorisation for POVMs, together with its consequences for joint measurability.

\section*{Acknowledgements}

The authors are thankful T. Perche for fruitful discussions.
LG is supported by the S\~ao Paulo Research Foundation (FAPESP) under grants 2016/01343-7 and 2018/04208-9. AB is partially supported
by CNPq.


\begin{thebibliography}{10}
\bibitem{stupidfootnote}
{From this point on we abuse the notation and write $\Pi_l$ both for the permutation map $\Pi_l:\{1,\ldots,n\}\rightarrow\{1,\ldots,n\}$ and for its matrix representation.}
\end{thebibliography}
\end{document}